\documentclass[10pt,conference]{IEEEtran}
 
\usepackage{research5}
\usepackage{psfrag}
\usepackage{epsfig}
\usepackage{graphicx}
\usepackage{multicol}
\usepackage{cite} 

\newtheorem{theorem}{Theorem}
\newtheorem{lemma}{Lemma}

\title{Private-Capacity Bounds for Bosonic Wiretap Channels}

\author{\authorblockN{Ligong Wang, \em Member IEEE\/\rm, Jeffrey
    H. Shapiro, \em Life Fellow IEEE\/\rm,\\ Nivedita Chandrasekaran, Gregory 
    W. Wornell, \em Fellow IEEE} \authorblockA{Research Laboratory of Electronics,
    Massachusetts Institute of Technology, 
    Cambridge, MA, USA\\
    \texttt{\{wlg, jhs, nivedita, gww\}@mit.edu}}}

\begin{document}

\maketitle

\begin{abstract}
We prove an upper bound on the private capacity of the single-mode
noiseless bosonic wiretap channel. Combined with a previous lower
bound, we obtain the low photon-number asymptotic expression for the
private capacity. We then show that the multiple-mode noiseless
bosonic wiretap channel is equivalent to parallel single-mode
channels, hence the single-mode bounds can be applied. Finally, we
consider multiple-spatial-mode propagation through atmospheric
turbulence, and derive a private-capacity lower bound that only
requires second moments of the channel matrix.

\end{abstract}

\section{Introduction}

In a variety of emerging applications, there is a need for secure
transmission over optical links.  In such settings, a natural approach
to providing security against computationally-unbounded attacks is to
exploit the physical layer, and the corresponding natural information
theoretic model for analysis is the basic bosonic wiretap channel.  In
practice, a variety of regimes are of interest.  While for optical
links photon efficiency (b/photon) has tended to be of greater
importance than spectral efficiency (b/s/Hz), there is growing
interest in the latter as well.  When the photon and spectral
efficiency requirements are simultaneously high, multiple spatial
modes are required \cite{GuhaDuttonShapiro11}.  Accordingly there is a need to more
fully understand the capacity of both single-mode and multiple-mode
bosonic wiretap channels in such regimes.  Moreover, in practice,
free-space propagation is strongly affected by turbulence, the effect
of which on private capacity is also not yet well understood.

In this paper, we first prove an upper bound on the private capacity
of the single-mode bosonic wiretap channel.  Combining our upper bound
with the previously-derived lower bound \cite{guhashapiroerkmen08}, we
obtain the single-mode private capacity's low photon-number asymptotic
behavior.  We then treat the multiple-mode bosonic wiretap channel,
obtaining results that tightly bracket the private capacity when both
high photon efficiency and high spectral efficiency are required.
Finally, we exploit convexity and majorization to obtain a lower bound
on the multiple-spatial-mode private capacity for the turbulent
channel that only requires second moments of the channel matrix, and
thus may be tight for near-field operation in which both high photon
efficiency and high spectral efficiency are obtained.

\section{Notation}

We use a lower-case letter like $x$ to denote a number, and an
upper-case letter like $X$ to denote a random variable (except for
some special cases, e.g., $C$ denotes the capacity). We use a boldface
lower-case letter like $\boldsymbol{x}$ or $\boldsymbol{\eta}$ to
denote a vector, and a boldface upper-case letter like
$\boldsymbol{X}$ or $\boldsymbol{H}$ to denote a random vector. We use
a font like $\mathsf{t}$ to denote a matrix, and a
corresponding upper-case letter like $\mathsf{T}$ to denote a random
matrix. Finally, we use a font like $\mathbb{A}$ to denote a Hilbert
space, a 
font like $\hat{a}$ to denote the annihilation operator on
$\mathbb{A}$, and $\hat{\rho}^{\mathbb{A}}$ to denote a density
operator on $\mathbb{A}$.

All logarithms in this paper are natural logarithms, and information
is measured in nats unless stated otherwise.


\section{The Single-Mode Channel}

\subsection{Channel Model and Previous Work}
Let
$\hat{a}$, $\hat{b}$, and $\hat{e}$ denote the annihilation operators
on the Hilbert spaces of Alice, Bob, and Eve, respectively. The
single-mode noiseless 
bosonic wiretap channel can be
described in the Heisenberg picture by the beam splitter relation
\begin{subequations}\label{eq:singlemode}
\begin{IEEEeqnarray}{rCl}
  \hat{b} & = & \sqrt{\eta} \hat{a}+\sqrt{1-\eta} \hat{v},\\
  \hat{e} & = & \sqrt{1-\eta} \hat{a}-\sqrt{\eta} \hat{v},
\end{IEEEeqnarray}
\end{subequations}
where $\eta\in[0,1]$, and where
$\hat{v}$ is the annihilation operator of the noise mode, which 
we assume 
to be in its vacuum state. Note that this is a \emph{worst-case}
model in the sense that we assume Eve can obtain all photons that do
not reach Bob. We impose an average-photon-number
constraint on the input
\begin{equation}
  \bra \hat{a}^\dag \hat{a} \ket \le \bar{n},\label{eq:constraint}
\end{equation}
where the expectation is averaged over all codewords.
Denote the classical private capacity of the channel
\eqref{eq:singlemode} under constraint \eqref{eq:constraint} by
$C_{\textnormal{P}}(\eta,\bar{n})$. It is shown in
\cite{guhashapiroerkmen08} that 
\begin{equation}\label{eq:lower}
  C_{\textnormal{P}}(\eta,\bar{n}) \ge L(\eta,\bar{n}),
\end{equation}
with
\begin{equation}\label{eq:L}
  L(\eta,\bar{n})=\begin{cases}g(\eta\bar{n})-g\left((1-\eta)\bar{n}\right),& \eta>1/2,\\0,&\textnormal{otherwise,}\end{cases}
\end{equation}
where
\begin{equation}\label{eq:g}
  g(x)\triangleq (1+x)\log(1+x)-x\log x,\quad x>0
\end{equation}
is the maximum entropy of a single-mode bosonic state whose expected
photon-number equals $x$, achieved by the \emph{thermal state}:
\begin{equation}
\hat{\rho}=  \sum_{n=0}^\infty \frac{x^n}{(x+1)^{n+1}} |n\ket\bra n|,
\end{equation}
where $|n\ket$ denotes the number state containing $n$ photons. 

It is conjectured in
\cite{guhashapiroerkmen08} that \eqref{eq:lower} holds with equality,
as a consequence of the conjectured ``Entropy Photon-Number
Inequality''.

As $\bar{n}$ tends to infinity, the lower bound \eqref{eq:lower} is
tight and agrees with the private-capacity formula derived in
\cite{wolfgarciagiedke06}:
\begin{equation}
  C_\textnormal{P}(\eta,\infty) = \max \left\{0,
    \log(\eta)-\log(1-\eta)\right\}.
\end{equation}

\subsection{An Upper Bound on $C_\textnormal{P}(\eta,\bar{n})$}

\begin{theorem}\label{thm:single}
  The classical private capacity $C_{\textnormal{P}}(\eta,\bar{n})$ is
  bounded by
  \begin{equation}\label{eq:upper}
    C_{\textnormal{P}}(\eta,\bar{n}) \le U(\eta,\bar{n}),
  \end{equation}
  where
  \begin{equation}\label{eq:U}
    U(\eta,\bar{n})\triangleq \begin{cases}
      g\left((2\eta-1)\bar{n}\right),& \eta>1/2,\\0,&
      \textnormal{otherwise.}\end{cases}
  \end{equation}
\end{theorem}

Before proving Theorem~\ref{thm:single}, we first prove a simple lemma
which says that $C_\textnormal{P}(\eta,\bar{n})$ is monotonic in $\eta$.

\begin{lemma}\label{lem:monotonicity}
  For any $1\ge \eta_1\ge\eta_2\ge0$ and any $\bar{n} > 0$,
  \begin{equation}\label{eq:monotonicity}
    C_\textnormal{P}(\eta_1,\bar{n}) \ge
    C_\textnormal{P}(\eta_2,\bar{n}).
  \end{equation}
\end{lemma}
\begin{proof}
  Let $\mathbb{B}_i$ and $\mathbb{E}_i$ denote the output Hilbert
  spaces of Bob and Eve, respectively, of the channel with
  transmissivity (from Alice to Bob) $\eta_i$, $i=1,2$. Observe that
  $\mathbb{B}_2$ is stochastically degraded from
  $\mathbb{B}_1$. Indeed, when we pass the state on
  $\mathbb{B}_1$ through a beam splitter of transmissivity
  $\eta_2/\eta_1$, we obtain a state that is identical to the one on
  $\mathbb{B}_2$. Therefore, a Bob having access to $\mathbb{B}_1$ can
  always pass his state through this beam splitter and then make the
  same measurement as a Bob having access to $\mathbb{B}_2$, thus he can
  do \emph{at least} as well as the latter. Similarly, $\mathbb{E}_1$ is
  stochastically degraded from $\mathbb{E}_2$, and an Eve having
  access to $\mathbb{E}_1$ can do \emph{at most} as well as an Eve
  having access to $\mathbb{E}_2$. Hence we obtain
  \eqref{eq:monotonicity}. 
\end{proof}

\begin{IEEEproof}[Proof of Theorem~\ref{thm:single}]
  By Lemma~\ref{lem:monotonicity}, we only need to prove the 
  case where $\eta > 1/2$. In this case the wiretap channel is
  stochastically degraded. To see this, we pass Bob's state through
  another beam splitter to obtain output modes with annihilation operators $\hat{e}'$ and $\hat{c}$ given by
  \begin{subequations}\label{eq:degraded}
  \begin{IEEEeqnarray}{rCl}
    \hat{e}'&=& \sqrt{\eta'}\hat{b}+\sqrt{1-\eta'}\hat{v}',\\
    \hat{c} &=& \sqrt{1-\eta'}\hat{b}-\sqrt{\eta'}\hat{v}',
  \end{IEEEeqnarray}
  \end{subequations}
  where
  $\eta'\triangleq (1-\eta)/\eta\in[0,1)
  $ as we assume $\eta>1/2$, 
  and $\hat{v}'$ is in its vacuum state.
  Then the states $\hat{\rho}^{\mathbb{E}}$ and
  $\hat{\rho}^{\mathbb{E}'}$ are identical 
  for any input state $\hat{\rho}^{\mathbb{A}}$. See
  Fig.~\ref{fig:degraded}. 
  \begin{figure}[h]
  \centering
  \psfrag{a}{$\hat{a}$}
  \psfrag{b}{$\hat{b}$}
  \psfrag{c}{$\hat{c}$}
  \psfrag{e}{$\hat{e}$}
  \psfrag{e'}{$\hat{e}'$}
  \psfrag{v}{$\hat{v}\colon |0\ket$}
  \psfrag{v'}{$\hat{v}'\colon |0\ket$}
  \psfrag{eta}{$\eta$}
  \psfrag{eta'}{$\eta'=(1-\eta)/\eta$}
  \includegraphics[width=0.3\textwidth]{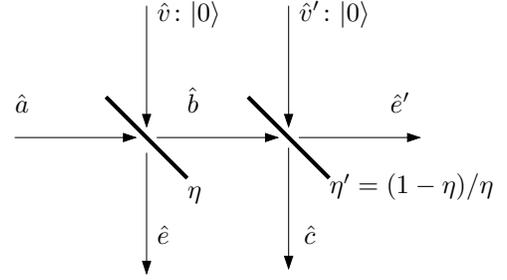}
  \caption{Illustration of the degraded wiretap channel.}
  \label{fig:degraded}
  \end{figure}

  We now prove \eqref{eq:upper} as follows:
  \begin{IEEEeqnarray}{rCl}
    C_\textnormal{P}(\eta,\bar{n}) & = & \max_{\bra
      \hat{a}^\dag\hat{a}\ket\le \bar{n}}
    \left[ S\left(\hat{\rho}^{\mathbb{B}}\right)-
    S\left(\hat{\rho}^\mathbb{E}\right)\right] \label{eq:11}\\ 
    & = & \max_{\bra
      \hat{a}^\dag\hat{a}\ket\le \bar{n}}
    \left[ S\left(\hat{\rho}^\mathbb{B}\right)-
    S\left(\hat{\rho}^{\mathbb{E}'}\right) \right]\label{eq:12}\\ 
    & = & \max_{\bra
      \hat{a}^\dag\hat{a}\ket\le \bar{n}}
    \left[ S\left(\hat{\rho}^\mathbb{B}\otimes |0\ket\bra
      0|^{\mathbb{V}'} 
    \right)-S\left(\hat{\rho}^{\mathbb{E}'}\right)
    \right] \label{eq:13} 
    \IEEEeqnarraynumspace \\
    & = & \max_{\bra
      \hat{a}^\dag\hat{a}\ket\le \bar{n}}
    \left[S\left(\hat{\rho}^{\mathbb{E}'\mathbb{C}}\right)-
    S\left(\hat{\rho}^{\mathbb{E}'}\right) \right] \label{eq:14}\\ 
    & \le & \max_{\bra
      \hat{a}^\dag\hat{a}\ket\le \bar{n}}
    \left[
    S\left(\hat{\rho}^{\mathbb{E}'}\right)+S\left(\hat{\rho}^\mathbb{C} 
    \right)- 
    S\left(\hat{\rho}^{\mathbb{E}'}\right) \right] \label{eq:15}\\
    & = & \max_{\bra
      \hat{a}^\dag\hat{a}\ket\le \bar{n}}
    S\left(\hat{\rho}^\mathbb{C}\right) \label{eq:16}\\
    & = & g\left((2\eta-1)\bar{n}\right), \label{eq:17}
  \end{IEEEeqnarray}
  where $S(\hat{\rho})= -{\rm Tr}[\hat{\rho}\log(\rho)]$ is the von Neumann entropy.
  The steps are justified as follows: \eqref{eq:11} follows from
  \cite[Theorem 2]{smith07}; \eqref{eq:12} from the fact that
  $\hat{\rho}^\mathbb{E}$ and $\hat{\rho}^{\mathbb{E}'}$ are
  identical; \eqref{eq:13} 
  because $|0\ket\bra 0|^{\mathbb{V}'}$ is a pure state; \eqref{eq:14}
  because the beam splitter
  \eqref{eq:degraded} is a unitary transformation from
  $\hat{\rho}^{\mathbb{BV}'}$ to $\hat{\rho}^{\mathbb{E}'\mathbb{C}}$;
  \eqref{eq:15} from the 
  subadditivity of von Neumann entropy; and \eqref{eq:17} because,
  according to the channel laws \eqref{eq:singlemode} and
  \eqref{eq:degraded},
  \begin{equation}
    \hat{c}=\sqrt{2\eta-1}
    \hat{a}+\sqrt{\frac{(2\eta-1)(1-\eta)}{\eta}} \hat{v}-
    \sqrt{\frac{1-\eta}{\eta}} \hat{v}', 
  \end{equation}
  so
  \begin{equation}
    \bra \hat{c}^\dag \hat{c} \ket = (2\eta-1) \bra \hat{a}^\dag
    \hat{a} \ket \le (2\eta -1)\bar{n}.
  \end{equation}
\end{IEEEproof}

\subsection{Analysis of the Bounds}

Combining the upper and lower bounds \eqref{eq:upper} and
\eqref{eq:lower} and letting $\bar{n}$ tend to zero, we obtain the
asymptotic expression for $C_{\textnormal{P}}(\eta,\bar{n})$ when
$\bar{n}$ is small.

\begin{theorem}\label{thm:asymptotic}
  The private capacity $C_{\textnormal{P}}(\eta,\bar{n})$ satisfies
  \begin{equation}
    C_{\textnormal{P}}(\eta,\bar{n}) =
    (2\eta-1)\bar{n}\log\frac{1}{\bar{n}}+O(\bar{n}),
  \end{equation}
  where $O(\bar{n})$ is a function of $\eta$ and $\bar{n}$ satisfying
  \begin{IEEEeqnarray}{rCl}
    \lefteqn{\eta\left(1+\log\frac{1}{\eta}\right)-
      (1-\eta)\left(1+\log\frac{1}{1-\eta} 
      \right) \le  \varliminf_{\bar{n}\downarrow 0}
    \frac{O(\bar{n})}{\bar{n}}}~~~~~~~~~~\nonumber\\ & \le &
    \varlimsup_{\bar{n}\downarrow 0} \frac{O(\bar{n})}{\bar{n}} \le
    (2\eta-1)\left(1+\log\frac{1}{2\eta-1}\right). \label{eq:difference}
    \IEEEeqnarraynumspace 
  \end{IEEEeqnarray}
\end{theorem}

Theorem~\ref{thm:asymptotic} shows that the \emph{photon efficiency},
$C_{\textnormal{P}}(\eta,\bar{n})/\bar{n}$, behaves like
$\log(1/\bar{n})$ plus some constant for small $\bar{n}$. We
numerically compare the upper and lower bounds \eqref{eq:upper} and 
\eqref{eq:lower} on the photon efficiency against $\bar{n}$ for
$\eta=0.7$ in Fig.~\ref{fig:nbar}, and against $\eta$ for
$\bar{n}=10^{-3}$ in Fig.~\ref{fig:eta}.
\begin{figure}[h!]
\centering
\psfrag{nbar}{$\bar{n}$}
\includegraphics[width=0.39\textwidth]{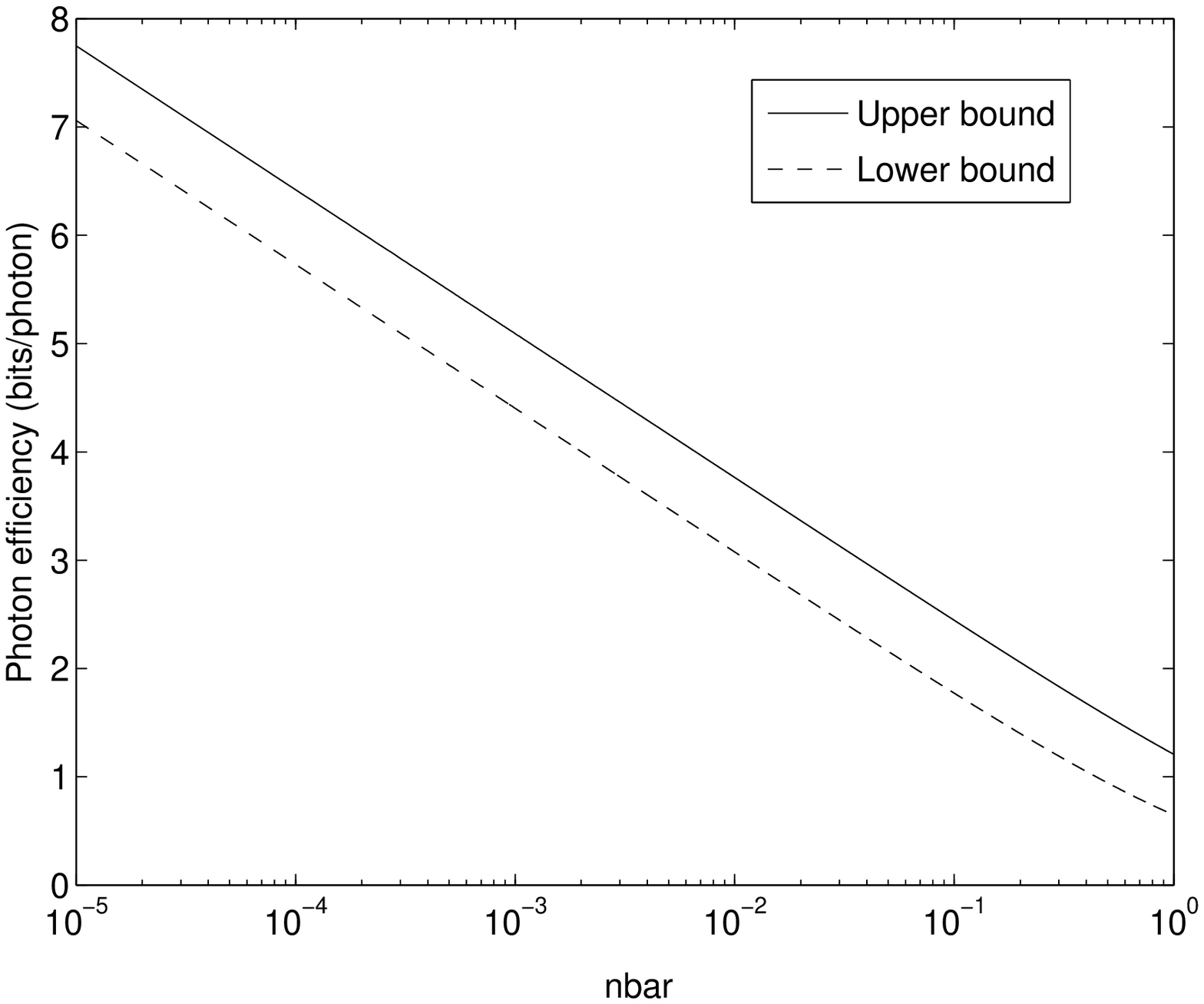}
\caption{Comparison of the upper and lower bounds on the photon
    efficiency (in bits per photon) computed from \eqref{eq:upper}
  and \eqref{eq:lower} for $\eta=0.7$.}
\label{fig:nbar}
\end{figure}

\begin{figure}[h!]
\centering
\psfrag{eta}{$\eta$}
\includegraphics[width=0.39\textwidth]{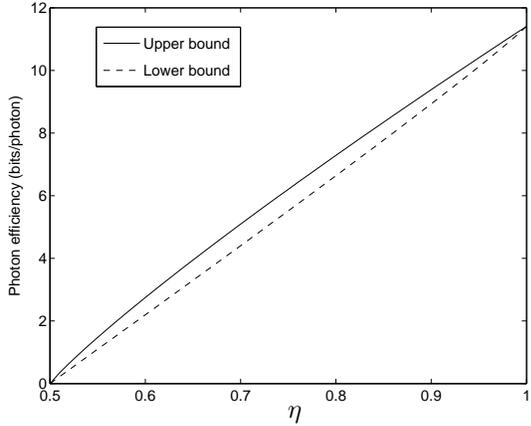}
\caption{Comparison of the upper and lower bounds on the photon
    efficiency (in bits per photon) computed from \eqref{eq:upper}
  and \eqref{eq:lower}, for $\bar{n}=10^{-3}$.}
\label{fig:eta}
\end{figure}

\section{The Multiple-Mode Channel}

\subsection{Channel Model}

Consider a multiple-mode noiseless bosonic wiretap channel in which Alice's,
Bob's, and Eve's modes are described by annihilation operators
$\{\hat{a}_1,\ldots,\hat{a}_m\}$, $\{\hat{b}_1,\ldots,\hat{b}_k\}$,
and $\{\hat{e}_1,\ldots,\hat{e}_l\}$, respectively. The channel law is a multiple-mode beam splitter relation, i.e., 
\begin{equation}\label{eq:MIMOME}
  \left( \begin{array}{c} \hat{b}_1 \\ \vdots \\ \hat{b}_k \\
      \hat{e}_1 \\ \vdots \\ \hat{e}_l \end{array} \right) =
  \left(\begin{array}{cc} \mathsf{t}_{ab} & \mathsf{t}_{vb} \\
      \mathsf{t}_{ae} & 
      \mathsf{t}_{ve} \end{array} \right) \left(\begin{array}{cc}
      \hat{a}_1 \\ 
      \vdots \\ \hat{a}_m \\ \hat{v}_1 \\ \vdots \\
      \hat{v}_{k+l-m} \end{array} \right),
\end{equation}
where
\begin{equation} \label{eq:deft}
  \mathsf{t} \triangleq  \left(\begin{array}{cc} \mathsf{t}_{ab} &
      \mathsf{t}_{vb} \\ \mathsf{t}_{ae} & 
      \mathsf{t}_{ve} \end{array} \right)
\end{equation}
is a unitary matrix, and where $\{\hat{v}_1,\ldots,\hat{v}_{k+l-m}\}$
are annihilation operators of vacuum-state noise modes.  Note that this
is again a worst-case model in which Eve obtains all photons that do not
reach Bob.  

\subsection{Simplification of Channel
  Model} \label{sub:simplification} 

The next theorem shows that any multiple-mode noiseless bosonic
wiretap channel is equivalent to a group of parallel
(i.e., noninterfering) single-mode channels.

\begin{theorem}\label{thm:equivalent}
  The channel \eqref{eq:MIMOME} is equivalent to a group of parallel
  single-mode channels:
  \begin{subequations}\label{eq:parallel}
  \begin{IEEEeqnarray}{rCl}
    \hat{b}_i' & = & \sqrt{\eta_i} \hat{a}_i' + \sqrt{1-\eta_i}
    \hat{v}_i',\\
    \hat{e}_i' & = & \sqrt{1-\eta_i} \hat{a}_i' - \sqrt{\eta_i}
    \hat{v}_i',
  \end{IEEEeqnarray}
  \end{subequations}
  where $i\in\{1,\ldots,m\}$, and where $\{\eta_1,\ldots,\eta_m\}$ are the 
  eigenvalues of $\mathsf{t}_{ba}^\dag \mathsf{t}_{ba}$.
\end{theorem}

\begin{proof}
  From the unitarity of the transition matrix $\mathsf{t}$ we have
  \begin{equation}
    \left(\begin{array}{cc} \mathsf{t}_{ab}^\dag &
      \mathsf{t}_{ae}^\dag \\ \mathsf{t}_{vb}^\dag & 
      \mathsf{t}_{ve}^\dag \end{array} \right) \cdot 
    \left(\begin{array}{cc} \mathsf{t}_{ab} &
      \mathsf{t}_{vb} \\ \mathsf{t}_{ae} & 
      \mathsf{t}_{ve} \end{array} \right) =
  \mathsf{1}^{(k+l)\times(k+l)},
  \end{equation}
  so 
  \begin{equation}
    \mathsf{t}_{ab}^\dag\mathsf{t}_{ab} +
    \mathsf{t}_{ae}^\dag\mathsf{t}_{ae} = \mathsf{1}^{m\times m}.
  \end{equation}
  This implies that $\mathsf{t}_{ab}^\dag\mathsf{t}_{ab}$ and
  $\mathsf{t}_{ae}^\dag\mathsf{t}_{ae}$ are simultaneously
  diagonalizable. More specifically, there exists a unitary matrix
  $\mathsf{v}$ such that
  \begin{IEEEeqnarray}{rCl}
    \mathsf{v}^\dag \mathsf{t}_{ab}^\dag \mathsf{t}_{ab} \mathsf{v} &
    = & \mathsf{d},\\
    \mathsf{v}^\dag \mathsf{t}_{ae}^\dag \mathsf{t}_{ae} \mathsf{v} &
    = & \mathsf{1}^{m\times m} - \mathsf{d},
  \end{IEEEeqnarray}
  where $\mathsf{d}$ is an $m\times m$ diagonal matrix whose
  diagonal terms are, by assumption, $\eta_1,\ldots,\eta_m$. Therefore
  the matrices $\mathsf{t}_{ab}$ and $\mathsf{t}_{ae}$ have the same
  right singular vectors, and their singular-value decompositions can
  be written as
  \begin{IEEEeqnarray}{rCl}
    \mathsf{t}_{ab} & = & \mathsf{u}_{ab} \mathsf{s}_{ab}
    \mathsf{v}^\dag,\\
    \mathsf{t}_{ae} & = & \mathsf{u}_{ae} \mathsf{s}_{ae}
    \mathsf{v}^\dag,
  \end{IEEEeqnarray}
  where $\mathsf{u}_{ab}$ and $\mathsf{u}_{ae}$ are unitary matrices,
  $\mathsf{s}_{ab}$ is a $k\times m$ diagonal matrix whose
  (nonzero) diagonal entries are (the nonzero elements of)
  $\{\sqrt{\eta_1},\ldots,\sqrt{\eta_m}\}$, and  $\mathsf{s}_{ae}$ is
  an $l\times m$ diagonal matrix whose 
  (nonzero) diagonal entries are (the nonzero elements of)
  $\{\sqrt{1-\eta_1},\ldots,\sqrt{1-\eta_m}\}$.
  
  Now we observe that $\mathsf{v}^\dag$ does not affect the private
  capacity of this channel. This is because Alice can perform
  $\mathsf{v}$ on the 
  input light modes that she prepared to cancel $\mathsf{v}^\dag$
  simultaneously for Bob and Eve. Hence we can always set
  $\mathsf{v}^\dag$ to be $\mathsf{1}^{m\times m}$ without affecting
  the private capacity. Similarly, $\mathsf{u}_{ab}$ and
  $\mathsf{u}_{ae}$ can be canceled by Bob and Eve, respectively, so
  they can also be set to identity matrices without changing the
  private capacity. We thus conclude that the private capacity of
  \eqref{eq:MIMOME} is the same as that of the parallel-mode channel
  \eqref{eq:parallel}. 
\end{proof}

\subsection{Capacity Results}
Denote the private capacity of the channel \eqref{eq:MIMOME} under
the average-photon-number constraint
\begin{equation}\label{eq:MIMOMEconstraint}
  \sum_{i=1}^m \bra \hat{a}_i^\dag \hat{a}_i \ket \le \bar{n}
\end{equation}
by $C_\textnormal{P}^\textnormal{M}(\mathsf{t},\bar{n})$. By
Theorem~\ref{thm:equivalent}, it equals the capacity of the channel
\eqref{eq:parallel} under constraint
\begin{equation}
  \sum_{i=1}^m \bra \hat{a}_i'^\dag \hat{a}_i' \ket \le \bar{n},
\end{equation}
which we denote by
$C_\textnormal{P}^\textnormal{M}(\boldsymbol{\eta},\bar{n})$ where 
$\boldsymbol{\eta}\triangleq
(\eta_1,\ldots,\eta_m)^\textnormal{T}$. We first show that 
$C_\textnormal{P}^\textnormal{M} (\boldsymbol{\eta},\bar{n})$ is
achievable by coding 
independently for each mode in \eqref{eq:parallel}.

\begin{theorem}
  Coding independently for each mode in \eqref{eq:parallel} is optimal:
  \begin{equation}\label{eq:independent}
    C_\textnormal{P}^\textnormal{M} (\boldsymbol{\eta},\bar{n})  = 
    \max_{\substack{\bar{n}_i\ge 0,\quad i=1,\ldots,m,\\ \sum_{i=1}^m
        \bar{n}_i = \bar{n}}} \sum_{i=1}^m
    C_\textnormal{P}(\eta_i,\bar{n}_i).
  \end{equation}
\end{theorem}

\begin{proof}
  Let $\boldsymbol{\eta}'$ be
  \begin{equation}
    \eta_i' = \begin{cases} \eta_i,&\eta_i> 1/2,\\ 1/2,& \eta_i\le
      1/2. \end{cases}
  \end{equation}
  By extending Lemma~\ref{lem:monotonicity} to the multiple-mode
  scenario, we have
  \begin{equation}
     C_\textnormal{P}^\textnormal{M} (\boldsymbol{\eta},\bar{n}) \le
     C_\textnormal{P}^\textnormal{M} (\boldsymbol{\eta}',\bar{n}).
  \end{equation}
  Next denote by
  $C_\textnormal{P}^\textnormal{M}
  (\boldsymbol{\eta}',\bar{\boldsymbol{n}})$, where 
  $\bar{\boldsymbol{n}} \triangleq
  (\bar{n}_1,\ldots,\bar{n}_m)^\textnormal{T}$, 
  the capacity of the parallel-mode channel with transmissivities
  $\boldsymbol{\eta}'$ and \emph{individual} photon-number constraints
  \begin{equation}
    \bra \hat{a}_i'^\dag\hat{a}_i' \ket \le \bar{n}_i,\quad
    i=1,\ldots,m,
  \end{equation}
  then
  \begin{equation}
    C_\textnormal{P}^\textnormal{M} (\boldsymbol{\eta}',\bar{n}) =
    \max_{\substack{\bar{\boldsymbol{n}}\colon \bar{n}_i\ge 0,\quad
        i=1,\ldots,m\\ \sum_{i=1}^m \bar{n}_i = \bar{n}}} 
    C_\textnormal{P}^\textnormal{M} (\boldsymbol{\eta}',\bar{\boldsymbol{n}}).
  \end{equation}
  To simplify $C_\textnormal{P}^\textnormal{M}
  (\boldsymbol{\eta}',\bar{\boldsymbol{n}})$, 
  note that each individual channel of this parallel-mode channel is
  stochastically degraded, so the private capacities of the individual
  channels are \emph{additive} \cite{smith07}:
  \begin{equation}
    C_\textnormal{P}^\textnormal{M}
    (\boldsymbol{\eta}',\bar{\boldsymbol{n}}) = \sum_{i=1}^n 
    C_\textnormal{P} (\eta_i',\bar{n}_i).
  \end{equation}
  We thus have
  \begin{IEEEeqnarray}{rCl}
    C_\textnormal{P}^\textnormal{M} (\boldsymbol{\eta},\bar{n}) &\le&
    \max_{\substack{\bar{n}_i\ge 0,\quad i=1,\ldots,m,\\ \sum_{i=1}^m
        \bar{n}_i = \bar{n}}} \sum_{i\colon \eta_i>1/2}
    C_\textnormal{P}(\eta_i',\bar{n}_i) \label{eq:cp1}\\ & = &
   \max_{\substack{\bar{n}_i\ge 0,\quad i=1,\ldots,m,\\ \sum_{i=1}^m
        \bar{n}_i = \bar{n}}} \sum_{i\colon \eta_i>1/2}
    C_\textnormal{P}(\eta_i,\bar{n}_i),\label{eq:independentU}
    \IEEEeqnarraynumspace 
  \end{IEEEeqnarray}
  where the equality follows because the optimal photon-number
  allocation is the same for the right-hand sides of both \eqref{eq:cp1} and 
  \eqref{eq:independentU}, which assigns zero photon to the
  modes where $\eta_i\le 1/2$ (i.e., where $\eta_i'= 1/2$).
  
  On the other hand, by coding independently, and using the optimal
  code for each mode, we can achieve the lower bound
  \begin{equation} \label{eq:independentL}
     C_\textnormal{P}^\textnormal{M} (\boldsymbol{\eta},\bar{n}) \ge 
    \max_{\substack{\bar{n}_i\ge 0,\quad i=1,\ldots,m,\\ \sum_{i=1}^m
        \bar{n}_i = \bar{n}}} \sum_{i\colon \eta_i>1/2}
    C_\textnormal{P}(\eta_i,\bar{n}_i).
  \end{equation}
  Combining \eqref{eq:independentU} and \eqref{eq:independentL} proves
  \eqref{eq:independent}. 
\end{proof}

Now it is straightforward to extend the upper and lower bounds
\eqref{eq:upper} and \eqref{eq:lower} to the multiple-mode case. In
particular, in the limit as $\bar{n}$ approaches zero, it is easy to
check that the optimal
photon-number allocation is the same for both the upper and the lower
bounds, and it sends all photons in the mode with the largest
transmissivity. We hence have the following asymptotic capacity
expression. 

\begin{theorem}
  The capacity of the channel \eqref{eq:MIMOME} under constraint
  \eqref{eq:MIMOMEconstraint} satisfies
  \begin{equation}
    C_\textnormal{P}^\textnormal{M} (\mathsf{t}, \bar{n}) =
    (2\eta_{\max}-1)\bar{n} \log \frac{1}{\bar{n}}+O(\bar{n}),
  \end{equation}
  where $\eta_{\max}$ is the largest eigenvalue of
  $\mathsf{t}_{ab}^\dag \mathsf{t}_{ab}$, and where the term
  $O(\bar{n})$ is at most linear in $\bar{n}$:
  \begin{equation}
    \varlimsup_{\bar{n}\downarrow 0} \left| \frac{O(\bar{n})}{\bar{n}}
    \right| < \infty.
  \end{equation}
\end{theorem}

As an example of our multiple-mode private capacity bounds, consider the use of $m=10^3$ high-transmissivity spatial modes with near-equal, near-unity eigenvalues, $(\eta_1,\ldots,\eta_m)^T$, as exist for $L$\,m vacuum-propagation at wavelength $\lambda$ between coaxial diameter-$D$ circular pupils satisfying $(\pi D^2/4\lambda L)^2 \gg m$ \cite{Slepian65}.  Figure~\ref{fig:spiecomm} shows that our results provide tight bounds on the photon efficiency and spectral efficiency for this example.
\begin{figure}[h!]
\centering
\includegraphics[width=0.39\textwidth]{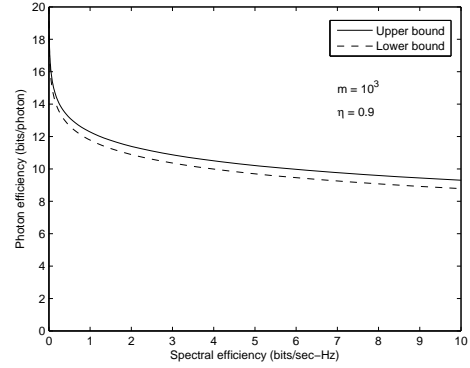}
\caption{Comparison of the upper and lower bounds for photon efficiency versus spectral efficiency for $m=10^3$ spatial modes each with $\eta = 0.9$.}
\label{fig:spiecomm}
\end{figure}

\section{Channels with Turbulence}

\subsection{Channel Model}
Consider a multiple-mode wiretap channel in which the transition matrix
$\mathsf{t}$ in \eqref{eq:deft} is replaced by a random matrix
$\mathsf{T}$. We assume a \emph{coherent} scenario 
where Alice and Bob know the realization of
$\mathsf{T}_{ab}$. Then, as discussed in
Section~\ref{sub:simplification}, they also know the other parts of
$\mathsf{T}$, namely, $\mathsf{T}_{ae}$, $\mathsf{T}_{vb}$ and
$\mathsf{T}_{ve}$ except for possible unitary transformations that
are irrelevant to capacity calculations. We impose the constraint
that the average number of transmitted photons in \emph{every} channel
use must not exceed $\bar{n}$, 
irrespective of the realization of $\mathsf{T}_{ab}$. Denote by
$\{H_1,\ldots,H_m\}$ the random 
eigenvalues of $\mathsf{T}_{ab}^\dag \mathsf{T}_{ab}$, then the
capacity of this channel can be expressed as
\begin{equation}
  C_\textnormal{P}^\textnormal{M} (\mathsf{T},\bar{n}) =
  \E{\max_{\substack{\bar{N}_i\ge 0,\quad i\in\{1,\ldots,m\} \\
        \sum_{i=1}^m \bar{N}_i = 
      \bar{n}}} C_\textnormal{P}(H_i,\bar{N}_i)}. 
\end{equation}
For near-field operation---wherein the turbulent channel will support multiple spatial modes with appreciable eigenvalues \cite{Shapiro74}---the exact distribution of $\mathsf{T}_{ab}$ is unavailable. Instead, we can only
compute the second-moment matrix
$
  \E{\mathsf{T}_{ab}^\dag \mathsf{T}_{ab}}
$.
Our goal in this section is to find good bounds on the private
capacity of 
the multiple-mode wiretap channel with turbulence expressed using
$\E{\mathsf{T}_{ab}^\dag \mathsf{T}_{ab}}$.

\subsection{Lower Bound on Private Capacity}
To derive a lower bound on the private capacity of this channel, we
need two lemmas.

\begin{lemma}\label{lem:convex}
  The single-mode lower bound $L(\eta,\bar{n})$ as defined in
  \eqref{eq:L} is convex in $\eta$ for $\eta\in[0,1]$ and for every
  $\bar{n}>0$. 
\end{lemma}
\begin{proof}
  Since $L(\eta,\bar{n})$ is the constant zero and is hence convex in
  $\eta$ for $\eta\in\left[0,\frac{1}{2}\right]$, 
  we only need to check
  convexity for $\eta\in\left(\frac{1}{2},1\right]$. In the latter region,
  \begin{equation}
  L(\eta,\bar{n}) = g(\eta\bar{n})-g((1-\eta)\bar{n}), \quad
  \eta\in\left(\frac{1}{2},1\right],
  \end{equation}
  and its second derivative with respect to $\eta$ can be computed:
  \begin{IEEEeqnarray}{rCl}
    \lefteqn{\frac{\d^2 L(\eta,\bar{n})}{d \eta^2}}~~~~ \nonumber\\ & = &
    \bar{n}^2 \left( 
      \frac{1}{(1+(1-\eta)\bar{n})(1-\eta)\bar{n}} -
      \frac{1}{(1+\eta\bar{n})\eta \bar{n}}
    \right)\IEEEeqnarraynumspace \\
    & \ge & 0,~~~~~~~~~~~~~~~~~~~~~~~~~~~~~~~\eta\in\left(
      \frac{1}{2},1\right]. 
  \end{IEEEeqnarray}
  Hence we conclude that $L(\eta,\bar{n})$ is convex in $\eta$ on
  $[0,1]$ and for every $\bar{n}>0$.
\end{proof}

\begin{lemma}
  The multiple-mode lower bound
  \begin{equation}\label{eq:bigL}
    L^\textnormal{M}(\boldsymbol{\eta},\bar{n}) \triangleq
    \max_{\substack{ \bar{n}_i\ge 0,\quad i\in\{1,\ldots,m\}\\ 
      \sum_{i=1}^m \bar{n}_i = 
      \bar{n}}} L(\eta_i,\bar{n}_i)
  \end{equation}
  is both convex and Schur-convex in $\boldsymbol{\eta}$.
\end{lemma}
\begin{proof}
  First note that $L^\textnormal{M} (\boldsymbol{\eta},\bar{n})$ is
  symmetric in the 
  elements of $\boldsymbol{\eta}$, hence convexity implies
  Schur-convexity \cite{marshallolkin79}. To prove convexity, consider
  any two vectors 
  $\boldsymbol{\eta}^a$, $\boldsymbol{\eta}^b$ and their mean
  $\boldsymbol{\eta}^c \triangleq
  (\boldsymbol{\eta}^a+\boldsymbol{\eta}^b)/2$. Suppose that
  $\bar{\boldsymbol{n}}^*$ achieves $L^\textnormal{M}
  (\boldsymbol{\eta}^c,\bar{n})$: 
  \begin{equation}
    L^\textnormal{M} (\boldsymbol{\eta}^c, \bar{n}) = \sum_{i=1}^m
    L(\eta_i^c,\bar{n}_i^*).
  \end{equation}
  We have
  \begin{IEEEeqnarray}{rCl}
    \lefteqn{\bigl( L^\textnormal{M} (\boldsymbol{\eta}^a, \bar{n}) +
      L^\textnormal{M} (\boldsymbol{\eta}^b, \bar{n}) \bigr)/2}~~~~~~
    ~~~~~~ 
    \nonumber\\ 
    & \ge & \left( \sum_{i=1}^m
    L(\eta_i^a,\bar{n}_i^*) + \sum_{i=1}^m
    L(\eta_i^b,\bar{n}_i^*) \right)/2 \label{eq:schur1}
  \IEEEeqnarraynumspace \\
    & \ge & \sum_{i=1}^m
    L(\eta_i^c,\bar{n}_i^*) \label{eq:schur2} \\
    & = & L^\textnormal{M} (\boldsymbol{\eta}^c,
    \bar{n}). \label{eq:schur3} 
  \end{IEEEeqnarray}
  Here: \eqref{eq:schur1} follows by lower-bounding the maxima over
  $\bar{\boldsymbol{n}}$ with the specific choice
  $\bar{\boldsymbol{n}}=\bar{\boldsymbol{n}}^*$; and \eqref{eq:schur2}
  by the 
  convexity of $L(\cdot,\bar{n})$ as in Lemma~\ref{lem:convex}. Hence
  $L^\textnormal{M}(\boldsymbol{\eta}, \bar{n})$ is convex in
  $\boldsymbol{\eta}$. 
\end{proof}

We are now ready to prove a lower bound on the private capacity of the
multiple-mode wiretap bosonic channel under turbulence which can be
expressed using $\E{\mathsf{T}_{ab}^\dag \mathsf{T}_{ab}}$.
\begin{theorem}\label{thm:turbulencelower}
  Let $\{\mu_1,\ldots,\mu_m\}$ denote the diagonal elements of
  $\E{\mathsf{T}_{ab}^\dag \mathsf{T}_{ab}}$, then
  \begin{equation}
    C_\textnormal{P}^\textnormal{M} (\mathsf{T},\bar{n}) \ge
    L^\textnormal{M} (\boldsymbol{\mu},\bar{n}),
  \end{equation}
  where $L^\textnormal{M} (\cdot,\cdot)$ is defined as in
  \eqref{eq:bigL}.
\end{theorem}

\textit{Remarks:} 
  As discussed in Section~\ref{sub:simplification}, the choice of
  basis for $\mathsf{T}$ does not affect the private capacity of our
  channel model, so Theorem~\ref{thm:turbulencelower} holds when
  $\boldsymbol{\mu}$ denotes the diagonal terms of $\E{\mathsf{T}_{ab}^\dag
    \mathsf{T}_{ab}}$ in \emph{any} basis. In particular, it holds if
  $\boldsymbol{\mu}$ denotes the eigenvalues of
  $\E{\mathsf{T}_{ab}^\dag \mathsf{T}_{ab}}$, and this choice of $\boldsymbol{\mu}$ provides the tightest bound obtainable in this manner.  Toward that end, the turbulence calculations from \cite{ChandrasekaranShapiro12} will permit this lower bound to be evaluated for transmitters that use focused-beam, Hermite-Gaussian, or Laguerre-Gaussian spatial modes.

\begin{proof}
  Let $\{M_1,\ldots,M_m\}$ denote the random diagonal elements of the
  random matrix $\mathsf{T}_{ab}^\dag \mathsf{T}_{ab}$.
  We have the following chain of inequalities:
  \begin{IEEEeqnarray}{rCl}
    C_\textnormal{P}^\textnormal{M} (\mathsf{T},\bar{n}) & \ge &
    \E{L^\textnormal{M} (\boldsymbol{H},\bar{n})} \\
    & \ge & \E{L^\textnormal{M}
      (\boldsymbol{M},\bar{n})} \label{eq:turblower1}\\ 
    & \ge & L^\textnormal{M}
    (\boldsymbol{\mu},\bar{n}). \label{eq:turblower2} 
  \end{IEEEeqnarray}
  Here: \eqref{eq:turblower1} follows by the Schur-convexity of
  $L^\textnormal{M} (\cdot,\bar{n})$ and the fact that the eigenvalues
  $\{H_1,\ldots,H_m\}$ majorize the diagonal elements
  $\{M_1,\ldots,M_m\}$; and \eqref{eq:turblower2} follows by the
  (normal) convexity of $L^\textnormal{M} (\cdot,\bar{n})$.
\end{proof}

\section*{Acknowledgements}
This research was supported by the DARPA InPho program under ARO Grant No. W911NF-10-1-0416, and by the NSF IGERT program Interdisciplinary Quantum Information Science and Engineering (iQuISE).


\begin{thebibliography}{2}

\bibitem{GuhaDuttonShapiro11}
S.~Guha, Z.~Dutton, and J.~H. Shapiro, ``On quantum limit of optical communications:  concatenated codes and joint-detection receivers,'' in \emph{Proc. IEEE Int.
  Symp. Inform. Theory}, Saint Petersburg, Russia, July 31--August 5, 2011.

\bibitem{guhashapiroerkmen08}
S.~Guha, J.~H. Shapiro, and B.~I. Erkmen, ``Capacity of the bosonic wiretap
  channel and the entropy photon-number inequality,'' in \emph{Proc. IEEE Int.
  Symp. Inform. Theory}, Toronto, Canada, July 6--11, 2008.

\bibitem{wolfgarciagiedke06}
M.~M. Wolf, D.~P\'erez-Garc\'ia, and G.~Giedke, ``Quantum capacities of bosonic
  channels,'' \em Phys. Rev. Lett\/\rm. {\bf 98,} 130501 2007.

\bibitem{smith07}
G.~Smith, ``The private classical capacity with a symmetric side channel and
  its application to quantum cryptography,'' \em Phys. Rev. A\/\rm, {\bf 78,} 022306 (2008).
  
\bibitem{Slepian65}
D.~Slepian, ``Analytic solution of two apodization problems,'' \em J. Opt. Soc. Am\/\rm. {\bf 55,} 1110-1115 (1965).  
  
\bibitem{Shapiro74}J.~H. Shapiro, ``Normal-mode approach to wave propagation in the turbulent atmosphere,'' \em Appl. Opt\/\rm. {\bf 13,} 2614--2619 (1974).

\bibitem{marshallolkin79}
A.~W. Marshall and I.~Olkin, \emph{Inequalities: Theory of Majorization and Its
  Applications}.\hskip 1em plus 0.5em minus 0.4em\relax Academic Press, 1979.
  
\bibitem{ChandrasekaranShapiro12}
N.~Chandrasekaran and J.~H. Shapiro, ``Turbulence-induced crosstalk in multiple-spatial-mode optical communication,'' submitted to \em CLEO 2012 Conference\/\rm, San Jose, CA, May 8--10, 2012.

\end{thebibliography}
\end{document}